\documentclass[conference]{IEEEtran}
\usepackage{etoolbox}
\makeatletter
\patchcmd{\@makecaption}
  {\scshape}
  {}
  {}
  {}
\makeatletter
\patchcmd{\@makecaption}
  {\\}
  {.\ }
  {}
  {}
\makeatother
\usepackage{amsfonts}
\usepackage{mathrsfs}
\usepackage{amsfonts}
\usepackage{amssymb}
\usepackage{graphicx}
\usepackage{epsfig}
\usepackage{psfrag}
\usepackage{amsmath}
\usepackage{array}
\usepackage{cases}
\usepackage{subfigure}
\usepackage{cite,graphicx,amsmath,amssymb,color}
\usepackage{anysize}
\usepackage{algorithmic}
\usepackage{algorithm}

\newtheorem{Thm}{Theorem}
\newtheorem{Lem}{Lemma}
\newtheorem{Cor}{Corollary}

\newtheorem{Prob}{Problem}
\newtheorem{Rem}{Remark}

\marginsize{13mm}{13mm}{19mm}{43mm}

\IEEEoverridecommandlockouts

\begin{document}

\title{Optimal Caching and User Association in  Cache-enabled Heterogeneous Wireless Networks}
\author{\authorblockN{Ying Cui, \ Fan Lai
%\thanks{Y. Cui was supported in part by the  National Science Foundation of China grant 61401272.}
}\authorblockA{Shanghai Jiao Tong University}\and\authorblockN{Stephen Hanly, Philip Whiting}\authorblockA{Macquarie University}
}

\maketitle

\begin{abstract} Heterogenous wireless networks (Hetnets) provide a powerful approach to meet the massive growth in traffic demands, but also impose a significant challenge  on  backhaul.  Caching at small base stations (BSs) and wireless small cell backhaul have been proposed as  attractive solutions to address this new challenge. In this paper, we consider  the optimal caching and user association  to minimize the total time to satisfy the average demands   in cached-enabled Hetnets with wireless backhaul. We formulate this problem as a mixed discrete-continuous optimization for given bandwidth and cache resources. First, we characterize the structure of the optimal solution. Specifically, we show that  the optimal caching is to store the most popular files at each pico BS, and the optimal user association  
%for both cached and uncached files 
has a threshold form. We also obtain the closed-form optimal solution in the homogenous scenario of pico cells. Then, we analyze the impact of bandwidth and cache resources on the minimum total time to satisfy the average demands. Finally, using numerical simulations, we verify the analytical results.
\end{abstract}

\section{Introduction}

The rapid proliferation of smart mobile devices has triggered
an unprecedented growth of the global mobile data traffic.
Heterogenous wireless networks (Hetnets) have been proposed as an effective way
to meet the dramatic traffic growth by deploying short range small base stations (BSs) together with traditional
macro BSs \cite{HetnetGhosh12}. Significant increase in network capacity is possible mainly because the small cells can operate simultaneously, providing better time or frequency reuse. However, this approach imposes a significant challenge of providing expensive high-speed backhaul links for connecting all the small BSs to the core network.
The backhaul capacity requirement can be enormously high
during peak traffic hours.

Caching at small BSs is a promising  approach to alleviate the backhaul capacity
requirement in Hetnets. Many existing
works have focused on optimal cache placement at small BSs, which is of critical importance in cache-enabled Hetnets.
For example, \cite{Shanmugam13} and \cite{Poularakis14} consider caching at small BSs in a single macro cell  with multiple small cells where the coverage areas of small cells are overlapping. File requests which cannot be satisfied locally at small BSs are served by the macro BS. Specifically, in \cite{Shanmugam13}, the authors consider the optimal caching design to minimize the expected delay for downloading uncached files from the macro BS.    In \cite{Poularakis14}, the authors consider the optimal caching to minimize the requests served by the macro BS.  The optimization problems  in \cite{Shanmugam13} and \cite{Poularakis14}  are NP-hard,  and simplified caching solutions are proposed with approximation guarantees.

Backhaul limitation is a critical problem in Hetnets.
In  \cite{EURASIP15Debbah,LiuYangICC16,Yang16,DebbahWiOpt14}, the authors consider caching at small BSs  in Hetnets  with backhaul constraints. Small BSs retrieve uncached files  via wireline backhaul from the core network and then transmit to local users. Thus, the service rate of uncached files at small BSs is also limited by the backhaul capacity.
Specifically,  \cite{EURASIP15Debbah,LiuYangICC16,Yang16} consider caching the most popular files at each small BS and focuses on analyzing the network performance. \cite{DebbahWiOpt14} considers least frequently used caching policy and studies the optimal user association.
Wireless backhaul is an attractive option for small BSs in Hetnets as it is easier to deploy and is more cost effective than fiber based backhaul. When wireless backhaul is considered, it is essential to optimally allocate  time or frequency resources between wireless backhaul for file retrievements and small  BSs for file transmissions.

To avoid inter-tier interference and increase the spatial reuse in Hetnets, the well known techniques of intercell interference coordination using almost blanking subframes (ABS slots)  and cell range expansion (CRE) have been developed. The joint optimization of ABS slots and CRE in a dynamic traffic scenario has been studied in \cite{HanlyJSAC15}, without considering cache resource and backhaul limitation.  In cache-enabled Hetnets with wireless backhaul, caching, ABS slots and CRE should be jointly optimized  under wireless backhaul constraints, in order to fully exploit the bandwidth and cache resources to improve  network capacity. Note that CRE is not considered in \cite{Shanmugam13}, \cite{Poularakis14}, and optimal caching   is not studied in \cite{EURASIP15Debbah} and \cite{DebbahWiOpt14}.
In this paper, we consider the joint optimization of caching, ABS slots and CRE  in  a cache-enabled Hetnet consisting of a single macro cell containing  multiple pico BSs with wireless backhaul. We formulate a mixed discrete-continuous optimization of the caching and  user association to minimize the total time to satisfy the average demands in a dynamic traffic scenario. Note that the  user association  reflects control of ABS slots and CRE. We show that the optimal caching is to store the most popular files at each pico BS, and the optimal  user association  for both cached and uncached files has a threshold form  resulting in two pico serving regions. The pico serving region for cached files  is larger and contains the pico serving region for uncached files. Both ranges are adaptive to the traffic density, cache resource and bandwidth resource.
The range difference comes from the difference in backhaul consumption.  We also obtain the closed-form optimal solution in the homogenous scenario of pico cells. Then, we analyze the impact of  bandwidth and cache resources on the minimum total time to satisfy the average demands. Finally, using numerical simulations, we verify the analytical results.

\section{System Model}

We consider a cache-enabled  Hetnet consisting of a single
macro cell containing $L\in \mathbb N_+$ pico BSs, as illustrated in Fig.~\ref{fig:model}.\footnote{The network topology and traffic model are similar to those in \cite{HanlyJSAC15}. However, \cite{HanlyJSAC15} does not  consider file popularity, cache and wireless backhaul.}  Let $\mathcal L\triangleq \{1, 2, \cdots, L\}$ denote the set of $L$ pico BSs.
%Each pico BS is connected to the macro BS via a wireless backhaul.
Let
$\mathcal K \subset\mathbb R^2$
denote the (compact) coverage area of the macro BS, while
$\mathcal K_l\subset \mathcal K$, $l\in \mathcal L$ denote the respective coverage areas of
the pico BSs. The pico BS coverage areas are assumed disjoint
such that a user at a location in $\mathcal K_l$ can obtain service from both
the macro BS and the pico BS $l$.
Let $\mathcal K_0\triangleq \mathcal K-\cup_{l=1}^L \mathcal K_l $
denote the set of locations not covered by any pico BS.
Any user at a location in  $\mathcal K_0$ will obtain all its support from the macro BS.

\begin{figure}
\begin{center}
 \includegraphics[width=5cm]{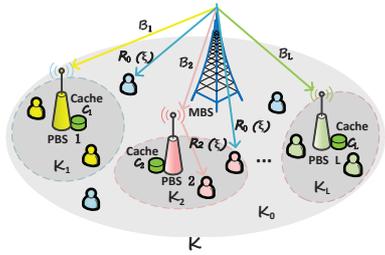}
  \end{center}
    \caption{\small{System model.}}
\label{fig:model}
\end{figure}

Let $\mathcal N\triangleq \{1, \cdots, N\}$ denote the set of $N$ files (contents) in the network. For
ease of illustration, we assume that all files have the same
size.\footnote{Files
of different sizes can be divided into chunks of the same length. Thus, the
results in this paper can be  extended to the case of different file sizes.} 
All the files are available at the macro BS. Each pico BS $l$ is equipped with a cache, which can store $C_l\in \{0,1,\cdots, N\}$ files.  Denote $\mathbf C\triangleq (C_l)_{l\in \mathcal L}$.   Let $s_{l,n}\in\{0,1\}$  represent the caching state for file $n$ at pico BS $l$, where $s_{l,n}=1$
if file $n$ is cached at pico BS $l$,  and $s_{l,n}=0$ otherwise. Denote $\mathbf s_l\triangleq (s_{l,n})_{n\in \mathcal N}$. Note that  $\mathbf s_l$ satisfies $\sum_{n=1}^N s_{l,n}\leq C_l$.
%We require
%\begin{align}
%\sum_{n=1}^N s_{l,n}\leq C_l, \ l\in \mathcal L.\label{eqn:cache-constr}
%\end{align}
For ease of illustration, we assume each pico BS $l$ can retrieve  uncached files from the macro BS via a wireless backhaul. Note that our formulation and solution hold when pico BSs retrieve uncached files from any connection point to the core network.

Users want to download files from the
network. We assume that the
file popularity distribution $\mathbf p\triangleq (p_n)_{n\in \mathcal N }$ is identical among all users and is known   apriori, where $p_n\in [0,1]$ is the popularity of file $n$ and $\mathbf p$ satisfies $\sum_{n\in \mathcal N}p_n=1$. In addition, without loss of generality, we assume $p_{1}\ge p_{2}\ldots\ge p_{N}$.
We assume that file requests arrive as a Poisson process with
arrival rate $\lambda_S$ files/sec and that an arrival is for file $n$ with probability $p_n$.  The locations of the arrivals are chosen independently at random according to a continuous density $\eta(d\xi)$ with support
on $\mathcal K$ and bounded uniformly away from 0. Users remain
fixed at their initial locations until they obtain their files. Hence, the expected number of requests per second at the vicinity
of a point $\xi\in\mathcal K$ is given by
$\lambda(d\xi) = \lambda_S \times\eta(d\xi)$.

%The probability that an arrival file request is allocated to picoregion
%Cl is given by
%? =
%
%C
%?(d?),  = 0, 1, 2, . . .L, (1)
%including the region C0 that is served only by the macro BS.
%The arrivals to each region are independent Poisson processes
%with rates ?S?,  = 0, 1, 2, . . . , L. The conditional density in
%region C is given by

The macro BS is assumed to use a much higher transmit
power than the pico BSs, to allow it to provide full coverage
of the region $\mathcal K$. We will therefore only consider control 
policies in which macro time and pico times are disjoint,
to avoid excessive interference at the users from the macro
BS when being served by pico BSs. During macro time, the macro BS can transmit to a user or a pico BS. On the other hand, pico BSs are spatially separated and
use much lower power. We will therefore consider control 
policies in which the pico BSs are allowed to operate simultaneously  during pico time. Let the time allocated to the pico cells be denoted by $f$ seconds.

The    transmission rate of a user is determined by the transmitting BS and
the user's  fixed location $\xi$.\footnote{We assume that the duration of a file transmission is long
enough to average the small-scale channel fading process.} All locations are in the macro BS
coverage area, and the corresponding    rate provided
by the macro BS to location $\xi\in\mathcal K$, if scheduled, is $R_0(\xi)$ file/sec/Hz.  If the location is within the coverage area of pico BS $l$, then an alternative rate provided by pico BS $l$ to location $\xi\in \mathcal K_l$, if scheduled,  is $R_l(\xi)$ file/sec/Hz. Assume that rates depend continuously on location with $0<R_{l,\min}\leq R_l(\xi)\leq R_{l,\max}$ for pico cell rates, where $\xi\in \mathcal K_l$ and $l\in \mathcal L$, and  with $0<R_{0,\min}\leq R_0(\xi)\leq R_{0,\max}$ for  the macro cell rates, where $\xi\in \mathcal K$. In addition, the    rate of the wireless backhaul from the macro BS  to pico BS $l$, if scheduled, is  $B_l$ file/sec/Hz. Since a pico BS has a much larger  receive antenna gain than a user,  we assume $R_0(\xi)<B_l$ for all $\xi\in \mathcal K_l$ and $l\in \mathcal L$. The total bandwidth is $W>0$ Hz.
%$W\in [\underline W, \bar W]$ Hz, where $\bar W>\underline W>0$.

Note that any user at a location in $\mathcal K_0$ is only served by the macro BS during macro time.
A user at a location in $\mathcal K_l$ can receive a file together from the macro BS and pico BS $l$ during macro time and pico time, respectively.
%The macro and pico BSs are
%scheduled at different times, so a user can receive a file from
%both types of BSs during macro and pico times, respectively.
Consider location $\xi\in \mathcal K_l$ in the coverage area of pico BS $l$. Let $x_{l,n}(\xi)\in [0,1]$ denote the fraction of file $n$ delivered by pico BS $l$ to location $\xi\in \mathcal K_l$ at rate $R_l(\xi)$ during pico time. If file $n$ is not stored at pico BS $l$, i.e., $s_{l,n}=0$, then this  $x_{l,n}(\xi)$ fraction of file $n$ has to be first delivered from the macro BS to pico BS $l$ at rate $B_l$ via the wireless backhaul during macro time; otherwise, the delivery of this  $x_{l,n}(\xi)$ fraction of file $n$  will not consume macro time. The remaining $1-x_{l,n}(\xi)$  fraction of file $n$ will be delivered by the macro BS to location $\xi\in \mathcal K_l$ at rate $R_0(\xi)$ during macro time.  Denote $\mathbf x_l(\xi)\triangleq (x_{l,n}(\xi))_{n\in \mathcal N}$ for all $\xi\in \mathcal K_l$ and $l\in \mathcal L$. Note that association  $\{\mathbf x_l(\xi):\xi\in \mathcal K_l, l\in \mathcal L,\}$ reflects control of ABS slots and CRE.

%The average data rate offered to a location will
%depend on the higher-layer controls: the cell association (pico
%or macro) and the time allocation offered by the selected BS(s).

%\section{Average Time Cost Minimization}\label{sec:timecostmin}
%
%In this section, we consider the average time cost minimization of the HetNet for  given bandwidth $W$ and storage size $\mathbf C$.

\section{Problem Formulation}

The time that the HetNet must be active in  order to satisfy given traffic demands
  is an important performance metric.
%  It is determined by the caching $\{\mathbf s_l:l\in \mathcal L\}$ and file scheduling $\{\mathbf x_l(\xi):\xi\in \mathcal K_l, l\in \mathcal L,\}$.
%Therefore, in this paper, we consider the optimal design and analysis  of the caching and file scheduling to  minimize the total amount of   time to satisfy the average file requests in the Hetnet.
Our goal is to find the optimal  pico time $f^*$, caching $\{\mathbf s_l^*:l\in \mathcal L\}$ and association  $\{\mathbf x_l^*(\xi):\xi\in \mathcal K_l, l\in \mathcal L\}$ to minimize the total  time that the HetNet must be active in  order to satisfy the average requests, under given system resources, i.e., bandwidth  resource $W$ and cache resource  $\mathbf C$.

\begin{Prob} [Optimal Caching and Scheduling]
\begin{align}
\tau^*=\min_{\substack{f,\{\mathbf s_l:l\in \mathcal L\},\\
\{\mathbf x_l(\xi):\xi\in \mathcal K_l, l\in \mathcal L\}}}\ & f+\sum_{n=1}^N\int_{\mathcal K_0}\frac{1}{WR_0(\xi)}p_{n} \lambda(d\xi)\nonumber\\
&+\sum_{l=1}^L \sum_{n=1}^N\int_{\mathcal K_l}\frac{1-x_{l,n}(\xi)}{WR_0(\xi)}p_{n} \lambda(d\xi)\nonumber\\
&+\sum_{l=1}^L \sum_{n=1}^N\int_{\mathcal K_l}\frac{(1-s_{l,n})x_{l,n}(\xi)}{WB_l}p_{n} \lambda(d\xi)\nonumber\\
s.t.\quad &\sum_{n=1}^N\int_{\mathcal K_l}\frac{x_{l,n}(\xi)}{WR_l(\xi)}p_{n} \lambda(d\xi)\leq f, \quad l\in \mathcal L\nonumber\\%\label{eqn:pico-f-constr}\\
& 0\leq x_{l,n}(\xi)\leq 1, \quad l\in \mathcal L, \ n\in \mathcal N\nonumber\\%\label{eqn:x-constr}\\
&s_{l,n}\in\{0,1\},\quad l\in \mathcal L, \ n\in \mathcal N\nonumber\\%\label{eqn:s-constr}\\
&\sum_{n=1}^N s_{l,n}\leq C_l, \quad l\in \mathcal L.\nonumber%\label{eqn:cache-constr}
\end{align}\label{prob:original}
\end{Prob}

Note that $\sum_{n=1}^N\int_{\mathcal K_0}\frac{1}{WR_0(\xi)}p_{n} \lambda(d\xi)=\int_{\mathcal K_0}\frac{1}{WR_0(\xi)}\lambda(d\xi)\triangleq \tau_0$ represents the macro time to satisfy the average requests from $\mathcal K_0$,\footnote{Note that $\tau_0$ is irrelevant to the optimization in Problem~\ref{prob:original}. We include it in the objective function for ease of the investigation of the impact of the system parameters  on the optimal total time $\tau^*$.} $\sum_{n=1}^N\int_{\mathcal K_l}\frac{1-x_{l,n}(\xi)}{WR_0(\xi)}p_{n} \lambda(d\xi)$ represents the macro time to directly satisfy the average requests from $\mathcal K_l$, while $\sum_{n=1}^N\int_{\mathcal K_l}\frac{(1-s_{l,n})x_{l,n}(\xi)}{WB_l}p_{n} \lambda(d\xi)$ represents the  macro time to satisfy the average requests from $\mathcal K_l$ indirectly, by first delivering  the uncached  files to pico BS $l$ via the wireless  backhaul.  $\sum_{n=1}^N\int_{\mathcal K_l}\frac{x_{l,n}(\xi)}{WR_l(\xi)}p_{n} \lambda(d\xi)$ represents the pico time to satisfy the average requests from $\mathcal K_l$. If pico BS $l$ carries all the average requests from $\mathcal K_l$, then it needs time $\bar f_l\triangleq \sum_{n=1}^N\int_{\mathcal K_l}\frac{1}{WR_l(\xi)}p_{n} \lambda(d\xi)= \int_{\mathcal K_l}\frac{1}{WR_l(\xi)}\lambda(d\xi)$. The maximum of such time over all pico BSs is $\bar f\triangleq \max_{l\in \mathcal L}\bar f_l$.

Problem~\ref{prob:original} is one in the calculus of variations over function $x_{l,n}(\xi)\in [0,1]$, $l\in \mathcal L$, $n\in \mathcal N$.  It  also has a continuous variable $f\geq 0$ and  $LN$ discrete variables $s_{l,n}\in \{0,1\}$, $l\in \mathcal L$, $n\in \mathcal N$. Note that the space of $x_{l,n}(\xi)\in [0,1]$, $l\in \mathcal L$, $n\in \mathcal N$ is compact, $f\geq 0$ can be taken no greater than $\bar f$, i.e., $f\in [0, \bar f]$, without affecting the optimality. In addition,  the space of $s_{l,n}\in \{0,1\}$, $l\in \mathcal L$, $n\in \mathcal N$ is finite, and the map from $x_{l,n}(\xi)\in [0,1]$, $l\in \mathcal L$, $n\in \mathcal N$ and  $f\in [0, \bar f]$ to the objective function in Problem~\ref{prob:original} is continuous for any $s_{l,n}\in \{0,1\}$, $l\in \mathcal L$, $n\in \mathcal N$.  Therefore, the minimum is achieved.

Using decomposition, Problem~\ref{prob:original}  can be equivalently transformed into the following master problem with $L$ subproblems.

\begin{Prob} [Master Problem of Problem~\ref{prob:original}]
\begin{align}
\tau^*=\min_{f}\quad &\tau(f)\nonumber\\
s.t.\quad & f\geq 0\nonumber
\end{align}
where $\tau(f)\triangleq  f+\tau_0
+\sum_{l=1}^L \tau_l(f)$
%\begin{align}
%\tau(f)\triangleq & f+\tau_0
%+\sum_{l=1}^L \tau_l(f)\label{eqn:def-tau-f}
%%\\
%%\tau_0\triangleq &\sum_{n=1}^N\int_{\mathcal K_0}\frac{1}{WR_0(\xi)}p_{n} \lambda(d\xi)= \int_{\mathcal K_0}\frac{1}{WR_0(\xi)} \lambda(d\xi)\label{eqn:def-tau-0}
%\end{align}
and $\tau_l(f)$  is given by the optimal value of Subproblem $l$ for given $f\geq 0$.\label{prob:master}
\end{Prob}
\begin{Prob} [Subproblem $l$ of Problem~\ref{prob:original}] For $f\geq 0$,
\begin{align}
\tau_l(f)=\min_{\substack{\mathbf s_l,\\ \{\mathbf x_l(\xi):\xi\in \mathcal K_l\}}}\quad &  \sum_{n=1}^N\int_{\mathcal K_l}\frac{1-x_{l,n}(\xi)}{WR_0(\xi)}p_{n} \lambda(d\xi)\nonumber\\
&+  \sum_{n=1}^N\int_{\mathcal K_l}\frac{(1-s_{l,n})x_{l,n}(\xi)}{WB_l}p_{n} \lambda(d\xi)\nonumber\\
s.t.\quad &\sum_{n=1}^N\int_{\mathcal K_l}\frac{x_{l,n}(\xi)}{WR_l(\xi)}p_{n} \lambda(d\xi)\leq f \label{eqn:pico-f-constr}\\
& 0\leq x_{l,n}(\xi)\leq 1, \quad   n\in \mathcal N\label{eqn:x-constr}\\
&s_{l,n}\in\{0,1\},\quad   n\in \mathcal N\label{eqn:s-constr}\\
&\sum_{n=1}^N s_{l,n}\leq C_l. \label{eqn:cache-constr}
\end{align}
\label{prob:subprob}
\end{Prob}

Note that $\tau_0$ denotes  the  macro time to satisfy the average requests from $\mathcal K_0$ and $\tau_l(f)$ denotes the optimal  macro time to satisfy the average requests from $\mathcal K_l$ directly and indirectly, given that the  pico time is no greater than $f$. Therefore, $\tau(f)$  represents the  total time to satisfy the average requests in the HetNet given that the   pico time  is no greater than $f$.

Problem~\ref{prob:subprob} for pico BS $l$ is a mixed discrete continuous optimization.  Similarly, it can be easily verified that the minimum is achieved. Problem~\ref{prob:master} is a continuous optimization problem over a single variable $f\geq 0$. Later, we shall show that Problem~\ref{prob:master} is convex.

\section{Optimality Properties}

\subsection{Optimal Solution}
In this part, we characterize the optimal solution to Problem~\ref{prob:original}.
We first characterize the optimal solution to  the subproblem for pico BS $l$ in Problem~\ref{prob:subprob} for any given $f\geq 0$. Then, we characterize the optimal solution to   Problem~\ref{prob:master}.

To solve Problem~\ref{prob:subprob}, we first solve the continuous relaxation of Problem~\ref{prob:subprob} where $s_{l,n}\in \{0,1\}$ is relaxed to $s_{l,n}\in [0,1]$. We show that  the optimal solution to the relaxed problem satisfies $s_{l,n}\in \{0,1\}$,  and hence it is also the optimal solution  to Problem~\ref{prob:subprob}.  Define $\rho_l(\xi,s)\triangleq \frac{R_l(\xi)}{R_0(\xi)}-\frac{(1-s)R_l(\xi)}{B_l}$ ($\xi\in \mathcal K_l$), $\underline{\rho_l}\triangleq \inf\{\rho_l(\xi,0):\xi\in \mathcal K_l\}$, $\bar{\rho_l}\triangleq \sup\{\rho_l(\xi,1):\xi\in \mathcal K_l\}$, $\mathcal A_l(\rho,s)\triangleq \{\xi\in \mathcal K_l: \rho_l(\xi,s)> \rho \}$, and $S_{C_l}\triangleq \sum_{n=1}^{C_l}p_n$. Note that  $\mathcal A_l(\rho,0)\subseteq\mathcal A_l(\rho,1)\subseteq \mathcal K_l$.  Suppose $\int_{\mathcal A_l(\rho,s)}d\xi$ is continuous in $\rho$.%\footnote{We omit all the proofs due to page limitation. Please refer to \cite{dropboxGC16CuiLai} for details.}
%In addition, $\rho_l(\xi,1)>\rho$ implies $R_l(\xi)>\rho R_0(\xi)$ and $\rho_l(\xi,0)>\rho$ implies $R_l(\xi)>\rho \frac{R_0(\xi)}{1-R_0(\xi)/B_l}$.
%The optimal solution to Problem~\ref{prob:subprob} is summarized in the following lemma.

\begin{Lem} For any given $f\geq 0$, the optimal solution to Problem~\ref{prob:subprob} is given by
\begin{align}
s_{l,n}^*=&
\begin{cases}1, & n=1,\cdots, C_l\\
0, & \text{otherwise}
\end{cases}\label{eqn:opt-s}\\
x_{l,n}^*(\xi, f)=&
\begin{cases}1, & \rho_{l}(\xi,s_{l,n}^*)> \rho_l(f)\\
0, & \text{otherwise}
\end{cases}\label{eqn:opt-x-rho}
\end{align}
where $\rho_l(f)$ satisfies
\begin{align}
&S_{C_l}\int_{\mathcal A_l(\rho_l(f),1)}\frac{1}{WR_l(\xi)} \lambda(d\xi)\nonumber\\
&+ (1-S_{C_l})\int_{\mathcal A_l(\rho_l(f),0)}\frac{1}{WR_l(\xi)} \lambda(d\xi)=f\label{eqn:rho-equation}
\end{align}
 if $f\leq\bar f_l$, and $\rho_l(f)=0$ otherwise.
%  if $f\leq\bar f_l \triangleq  S_{C_l}\int_{\mathcal A_l(0,1)}\frac{1}{WR_l(\xi)} \lambda(d\xi)+(1-S_{C_l})\int_{\mathcal A_l(0,0)}\frac{1}{WR_l(\xi)} \lambda(d\xi)$, and $\rho_l(f)=0$ otherwise.
% \footnote{Note that $\mathcal A_l(0,1)=\mathcal K_l$ and $\mathcal A_l(0,0)=\mathcal K_l$, as $0<R_{l,\min}\leq R_l(\xi)\leq R_{l,\max}$ and $R_0(\xi)<B_l$ for all  $\xi\in \mathcal K_l$ and $l\in \mathcal L$ and  $0<R_{0,\min}\leq R_0(\xi)\leq R_{0,\max}$  for all $\xi\in \mathcal K$.}
  In addition, the optimal value of Problem~\ref{prob:subprob} is given by
\begin{align}
\tau_l(f)=&S_{C_l}\int_{\mathcal K_l-\mathcal A_l(\rho_l(f),1)}\frac{1}{WR_0(\xi)} \lambda(d\xi)\nonumber\\
&+(1-S_{C_l})\int_{\mathcal K_l-\mathcal A_l(\rho_l(f),0)}\frac{1}{WR_0(\xi)} \lambda(d\xi)\nonumber\\
&+(1-S_{C_l})\int_{\mathcal A_l(\rho_l(f),0)} \frac{1}{WB_l}\lambda(d\xi).\label{eqn:subp-opt}
\end{align}\label{Lem:solution-f}
\end{Lem}
\begin{proof}
Please refer to Appendix A.
\end{proof}

\begin{Rem} [Optimal Structure] For any given pico time $f\geq 0$, the optimal caching is to cache the most popular $C_l$ files. Thus, $S_{C_l}$ reflects the cache hit probability. Given pico time $f\geq 0$, the optimal  association  $x_{l,n}^*(\xi, f)$ takes the threshold form where the threshold $\rho_l(f)$ depends on the pico time  $f$. Note that $\rho_l(f)$ also depends on $W$ and $C_l$, which are assumed to be fixed for now. The impact of   $W$ and $C_l$ on $\rho_l(f)$ will be studied later in Section~\ref{subsec:WC}. Specifically, for any cached file $n=1,\cdots, C_l$, the optimal   association   $x_{l,n}^*(\xi, f)=1$ if $\rho_{l}(\xi, 1)> \rho_l(f)$;  for any uncached file $n=C_l+1,\cdots, N$, the optimal  association  $x_{l,n}^*(\xi, f)=1$ if $\rho_{l}(\xi, 0)> \rho_l(f)$. In other words, the threshold $\rho(f)$ determines two pico serving regions $\mathcal A_l(\rho_l(f),1)\subseteq \mathcal K_l$ and $\mathcal A_l(\rho_l(f),0)\subseteq \mathcal K_l$ for the cached files and uncached files, respectively, where $\mathcal A_l(\rho_l(f),0)\subseteq \mathcal A_l(\rho_l(f),1)$.  For any cached file $n=1,\cdots, C_l$, the optimal  association  $x_{l,n}^*(\xi, f)=1$ if $\xi\in \mathcal A_l(\rho_l(f),1) $;  for any uncached file $n=C_l+1,\cdots, N$, the optimal  association   $x_{l,n}^*(\xi, f)=1$ if $ \xi\in \mathcal A_l(\rho_l(f),0)$. Each pico BS is more willing to serve requests for cached files, as no macro time  is consumed  for fetching these files from the macro BS via wireless backhaul.
\end{Rem}

%\begin{figure}
%\begin{center}
% \includegraphics[width=6cm]{fig/set.eps}
%  \end{center}
%    \caption{\small{Two pico serving regions.}}
%\label{fig:set}
%\end{figure}

From Lemma~\ref{Lem:solution-f}, we have the following corollary.
\begin{Cor} When $f=0$ (no pico time), we have $x_{l,n}^*(\xi, f)=0$ for all $\xi\in \mathcal K_l$ and
$\tau_l(f)=\int_{\mathcal K_l}\frac{1}{WR_0(\xi)} \lambda(d\xi)$. When $f\geq \bar f_l$, we have $x_{l,n}^*(\xi, f)=1$ for all $\xi\in \mathcal K_l$ and  $\tau_l(f)=(1-S_{C_l})\int_{\mathcal K_l} \frac{1}{WB_l}\lambda(d\xi)$. \label{Cor:solu-f0fbar}
\end{Cor}

Note that $\bar{f_l}$ can be interpreted as the largest pico time needed for pico BS $l$ to satisfy  the average requests from $\mathcal K_l$. When $f\in[0,\bar{f_l}]$, $\rho_l(f)$ satisfying \eqref{eqn:rho-equation}  is strictly decreasing  in $f$.
%\textcolor{red}{and has a unique finite derivation $\rho_l'(f)$ for almost all $f\in[0,\bar{f_l}]$.\footnote{Note that  $\rho_l(f)$  is strictly decreasing, and hence has bounded variation. By Theorem 34.3 [pp. 205, Integration xx], we know that $\rho_l(f)$    has a unique finite derivation $\rho_l'(f)$ for almost all $f\in[0,\bar{f_l}]$.}}
Thus, when $f\in[0,\bar{f_l}]$, $\tau_l(f)$ given by \eqref{eqn:subp-opt} is strictly decreasing   in $f$. When $f>\bar{f_l}$, $\rho_l(f)=0$ and  $\tau_l(f)=(1-S_{C_l})\int_{\mathcal K_l} \frac{1}{WB_l}\lambda(d\xi)$, which is the  average macro time to satisfy the average requests from $\mathcal K_l$ indirectly, %by delivering  the files from the macro BS to pico BS $l$ via the wireless backhaul,
and  does not change with $f$. In addition, by the structure of Problem~\ref{prob:subprob}, we can easily see that $\tau_l(f)$ is convex over $f$.

Next, we characterize the optimal solution to  Problem~\ref{prob:master}. %By definition, $\rho_l(0)=\bar{\rho_l}$ and $\rho_l(\bar{f_l})=\underline{\rho_l}$.
Based on the properties of $\tau_l(f)$ discussed above, we know that  $\tau(f)= f+\tau_0
+\sum_{l=1}^L \tau_l(f)$ is strictly  convex in $f$ over $[0,\bar{f}]$, and is strictly increasing in $f$ when $f>\bar f$. 
%\textcolor{red}{Hence, $\tau(f)$ has a unique optimum at $f^*\in [0, \bar f]$.}
%\textcolor{red}{Suppose $\tau_l(f)$ is differentiable over $f\in[0,\bar{f_l}]$ for all $l\in \mathcal L$.}
Then, we have the following lemma.
%First, we have the following lemma.
%\begin{Lem} When $f\geq 0$,  $\tau_l'(f)=-\rho_l(f)$.\label{Lem:tau-f-derivative}
%\end{Lem}
%Based on Lemma~\ref{Lem:tau-f-derivative}, we characterize the optimal solution to  Problem~\ref{prob:master} below.
%

\begin{Lem} The optimal solution $f^*$ to Problem~\ref{prob:master} exists and is unique. In addition, $f^*$ satisfies:
if $\sum_{l=1}^L \bar{\rho_l}\leq 1$, $f^*=0$ (no pico time) is optimal; $\sum_{l=1}^L \rho_l(\bar f) \geq 1$, $f^*=\bar f$ (all pico time) is optimal; otherwise, $f^*\in (0,\bar f)$ satisfies $\sum_{l=1}^L \rho_l(f^*)=1$, where $\rho_l(f)$ is given by Lemma~\ref{Lem:solution-f}. \label{Lem:opt-f}
\end{Lem}
\begin{proof}
Please refer to Appendix B.
\end{proof}

Based on Lemma~\ref{Lem:solution-f} and Lemma~\ref{Lem:opt-f}, we characterize  the optimal solution to Problem~\ref{prob:original} in the following theorem.

\begin{Thm} The optimal solution to Problem~\ref{prob:original} is given by
\begin{align}
s_{l,n}^*=&
\begin{cases}1, & n=1,\cdots, C_l\\
0, & \text{otherwise}
\end{cases}\label{eqn:opt-s-final}\\
x_{l,n}^*(\xi)=&
\begin{cases}1, & \rho_{l}(\xi,s_{l,n}^*)> \rho_l(f^*)\\
0, & \text{otherwise}
\end{cases}\label{eqn:opt-x-rho-final}
\end{align}
where $f^*$ given by Lemma~\ref{Lem:opt-f}  satisfies $S_{C_l}\int_{\mathcal A_l(\rho_l(f^*),1)}\frac{1}{WR_l(\xi)} \lambda(d\xi)
+ (1-S_{C_l})\int_{\mathcal A_l(\rho_l(f^*),0)}\frac{1}{WR_l(\xi)} \lambda(d\xi)\leq f^*$ for all $ l\in \mathcal L$.
%\begin{align}
%&S_{C_l}\int_{\mathcal A_l(\rho_l(f^*),1)}\frac{1}{WR_l(\xi)} \lambda(d\xi)\nonumber\\
%&+ (1-S_{C_l})\int_{\mathcal A_l(\rho_l(f^*),0)}\frac{1}{WR_l(\xi)} \lambda(d\xi)\leq f^*, \ l\in \mathcal L\nonumber
%\end{align}
In addition, the optimal value to Problem~\ref{prob:original} is given by
$\tau^*=f^*+\tau_0+\sum_{l=1}^L\tau_l(f^*)$, where $\tau_l(f)$ is given by Lemma~\ref{Lem:solution-f}.
 \label{Thm:opt-soul}
\end{Thm}

\subsubsection{Homogenous Scenario}

Now, we  consider the homogenous scenario across all the pico cells.  Specifically, in this scenario, we have $C_l=C$, $B_l=B$, $\bar{\rho_l}=\bar \rho$, $\underline{\rho_l}=\underline \rho$,  $\bar f_l=\bar f= \int_{\mathcal K_l}\frac{1}{WR_l(\xi)} \lambda(d\xi)$,  and $ \rho_l(f)= \rho(f)$ for all $l\in \mathcal L$, and $ \int_{\mathcal A_l(\rho,s)}\frac{1}{WR_l(\xi)} \lambda(d\xi)$ is the same for all $l\in \mathcal L$. By Lemma~\ref{Lem:opt-f}, we have the following corollary.

\begin{Cor} In the homogenous case, the optimal solution $f^*$ to Problem~\ref{prob:master} satisfies:
if $\bar{\rho}\leq \frac{1}{L}$, $f^*=0$ (no pico time) is optimal;
%and $\tau^*=\int_{\mathcal K}\frac{1}{WR_0(\xi)} \lambda(d\xi)$;
$\underline{\rho}\geq \frac{1}{L}$, $f^*=\bar f$ (all pico time) is optimal;
%and $\tau^*=\int_{\mathcal K_0}\frac{1}{WR_0(\xi)} \lambda(d\xi)+\int_{\mathcal K_l}\frac{1}{WR_l(\xi)} \lambda(d\xi)+(1-S_{C})\sum_{l=1}^L\int_{\mathcal K_l}\frac{1}{WB_l)} \lambda(d\xi)$;
otherwise,  $\rho(f^*)=\frac{1}{L}$ and $f^*\in (0,\bar f)$ is given by $f^*= S_{C}\int_{\mathcal A_l(\frac{1}{L},1)}\frac{1}{WR_l(\xi)} \lambda(d\xi)+ (1-S_{C})\int_{\mathcal A_l(\frac{1}{L},0)}\frac{1}{WR_l(\xi)} \lambda(d\xi)
$.
%\begin{align}
%f^*=& S_{C_l}\int_{\mathcal A_l(\frac{1}{L},1)}\frac{1}{WR_l(\xi)} \lambda(d\xi)\nonumber\\
%&+ (1-S_{C_l})\int_{\mathcal A_l(\frac{1}{L},0)}\frac{1}{WR_l(\xi)} \lambda(d\xi)\label{eqn:rho-equation-sym}
%\end{align}
\label{Cor:opt-f-sym}
\end{Cor}

Note that in the homogenous case,  the optimal threshold $\rho(f^*)=\frac{1}{L}$  ($\underline{\rho}< \frac{1}{L}<\bar{\rho}$)  reflects the resource  reuse with reuse factor $L$ in the Hetnet with $L$ pico BSs which are spatially separated and can be operated at the same time without mutual interference. Interestingly, the optimal threshold $\rho(f^*)=\frac{1}{L}$ no longer depends on $W$ and $\mathbf C$. Thus, the two pico serving regions for cached and uncached files in each pico cell do not change with $W$ and $\mathbf C$. In the homogenous scenario, we can directly obtain the closed-form optimal solution from Theorem~\ref{Thm:opt-soul} and Corollary~\ref{Cor:opt-f-sym}, without solving $\rho_l(f)$ from \eqref{eqn:rho-equation}.

\subsection{Impact of Bandwidth and Cache Size}\label{subsec:WC}
Problem~\ref{prob:original} is for given $W$ and $\mathbf C$. Thus, we can also write $f^*$, $\tau^*$, $\tau_l$, $\rho(f)$, $\bar f_l$ and $\bar f$ as $f^*(W,\mathbf C)$, $\tau^*(W,\mathbf C)$, $\tau_l(W, C_l)$, $\rho_l(f,W,C_l)$, $\bar f_l(W)$ and $\bar f(W)$, respectively.
From \eqref{eqn:rho-equation}, we can easily observe that  $\rho_l(f,W,C_l)$ increases in $C_l$ and decreases in $W$, when $f\leq \bar f_l(W) $. Thus, the optimal caching adapts to $\mathbf C$, while the optimal  association   adapts to $W$ and $\mathbf C$, via the optimal threshold $\rho_l(f^*(W,\mathbf C),W,C_l)$. In the following, we study how the optimal performance $\tau^*(W,\mathbf C)$ changes with system resources  $W$ and $\mathbf C$. 
We consider continuous relaxation of $C_l\in\{0,\cdots, N\}$ to $C_l\in[0,N]$. Define $S(C_l)=\sum_{n=1}^{\lfloor C_l\rfloor}p_n+(C_l-\lfloor C_l\rfloor)p_{\lceil C_l\rceil}$. Note that $S(C_l)=S_{C_l}$ when $C_l\in\{0,\cdots, N\}$. Replace $S_{C_l}$
with $S(C_l)$ in the corresponding expressions for $C_l\in\{0,\cdots, N\}$  when $C_l\in[0,N]$.
First, for all $W>0$ and $C_l\in (0,N)$, we have the following result.
\begin{Lem} When $f\in (0,\bar f_l(W))$, we have
\begin{align}
%&\frac{\partial\tau_l(f,W,C_l)}{\partial f}
%=-\rho_l(f, W, C_l)< 0\nonumber\\
&\frac{\partial\tau_l(f,W,C_l)}{\partial W}=-\frac{1}{W}\tau_l(C_l, W,f)-\frac{1}{W}\rho_l(f,W,C_l)f
<0\label{eqn:partial-W}\\
&\frac{\partial\tau_l(f,W,C_l)}{\partial C_l}\nonumber\\
=&p_{\lceil C_l\rceil}\int_{\mathcal A_l(\rho_l(f,W,C_l),1)}\left(\frac{\rho_l(f,W,C_l)}{WR_l(\xi)}-\frac{1}{WR_0(\xi)}  \right)\lambda(d\xi)\nonumber\\
&-p_{\lceil C_l\rceil}\int_{\mathcal A_l(\rho_l(f,W,C_l),0)}\left(\frac{\rho_l(f,W,C_l)}{WR_l(\xi)}-\frac{1}{WR_0(\xi)}  \right)\lambda(d\xi)\nonumber\\
&-p_{\lceil C_l\rceil}\int_{\mathcal A_l(\rho_l(f,W,C_l),0)}\frac{1}{WB_l} \lambda(d\xi)
<0.\label{eqn:partial-S}
\end{align}
When $f> \bar f_l(W)$, we have
\begin{align}
%\frac{\partial\tau_l(f,W,C_l)}{\partial f}=0\nonumber\\
\frac{\partial\tau_l(f,W,C_l)}{\partial W}=&-\frac{1}{W}\tau_l(C_l, W,f)
<0\label{eqn:partial-W-}\\
\frac{\partial\tau_l(f,W,C_l)}{\partial C_l}
=&-p_{\lceil C_l\rceil}\int_{\mathcal K_l} \frac{1}{WB_l}\lambda(d\xi)
<0.\label{eqn:partial-S-}
\end{align}
\label{Lem:derivative}
\end{Lem}

\begin{proof}
Please refer to Appendix C. 
\end{proof}
%Lemma~\ref{Lem:derivative} indicates how fast $\tau_l(f,W,C_l)$  decreases with $W$ and $C_l$ for given $f>0$.
Based on Lemma~\ref{Lem:derivative}, we know that  for all $W>0$ and $C_l\in (0,N)$, $l\in \mathcal L$,  when $f>0$, we can obtain $\frac{\partial\tau(f,W,\mathbf C)}{\partial W}<0$ and $\frac{\partial\tau(f,W,\mathbf C)}{\partial C_l}<0$, indicating  how fast $\tau(f,W,\mathbf C)$  decreases with $W$ and $\mathbf C$ for given $f>0$. Therefore, we can show that   as $W$ or $\mathbf C$ increases, the minimum total  time decreases.
\begin{Thm} For all $W, W'
>0$ and $C_l, C_l'\in (0,N)$, $l\in \mathcal L$, if $W\geq W'$ and $\mathbf C\succeq \mathbf C'$, then $\tau^*(W,\mathbf C)\leq \tau^*(W',\mathbf C')$, where the equality holds if and only if $W=W'$ and $\mathbf C= \mathbf C'$.\label{Thm:W-C-time}
\end{Thm}

\begin{proof}
Please refer to Appendix D.
\end{proof}

\section{Numerical Results}

\begin{table}[!t]
\centering
\caption{Simulation parameters.}\label{tab:aTable}
\scriptsize{\begin{tabular}{|c|c|c|c|}
\hline
Macro Tx. Power& 46 dBm & Pico Tx. Power & 30dBm\\
\hline
Macro-MS Ant. Gain & 14dBi & Pico-MS Ant. Gain &5dBi\\
\hline
Macro-Pico Ant. Gain & 17dBi & Noise Power &-104 dBm\\
\hline
\multicolumn{4}{|c|}{Macro pathloss (in dB) = 128.1+37.6$\log_{10}$(d/1000),$d>35$ m}\\
\hline
\multicolumn{4}{|c|}{Pico pathloss (in dB) = 140.7+36.7$\log_{10}$(d/1000),$d>10$ m}\\
\hline
\end{tabular}}\label{tab}
\end{table}

In this section, we illustrate the analytical results
via numerical examples.
Consider a circular macrocell, with three pico BSs ($L = 3$),
each of which is deployed at the centre of a circular hotspot. A hotspot is a region with higher user density
(explained further below). The macro cell has radius 1 km, and
each hotspot  has radius 150 m. The macro BS is located at the
origin and the pico BSs are located at  $(-339,741),(218,-230),(561,-457)$ in the heterogenous scenario and at $(450,0),(-225,-390), (-225,390)$ in the homogenous scenario.
In our simulations, all users are assigned to
a pico BS, i.e., there are no macro only users ($\mathcal K_0=\emptyset$).
Pico assignment is made according to the nearest (strongest)
pico, forming
%It follows that the pico coverage regions, $\mathcal K_1,\mathcal K_2, \mathcal K_3$ are
%determined by dividing the macro coverage area according to
Voronoi regions $\mathcal K_1,\mathcal K_2, \mathcal K_3$. The hotspot probabilities for the three pico cells are given by 0.4, 0.25, 0.15 (heterogenous scenario)  and 0.8/3, 0.8/3, 0.8/3 (homogeneous scenario). An arriving file is assigned to a given hotspot
according to its probability independently of other files.
There is a chance of 0.2 that an arriving file falls
outside any hotspot. In this case, the corresponding file is
assigned to the non-hotspot area. Once the region of an arriving
file is determined, the actual location is chosen uniformly
at random, with the exception that no mobile is placed within
10 m of any pico BS or within 35 m from the macro BS. We choose $N=1000$, file size $D=4$ Mbits and $\lambda_s=1$ file/sec.  We assume the file popularity follows Zipf distribution, i.e., $p_n=\frac{n^{-\gamma}}{\sum_{n\in \mathcal N}n^{-\gamma}}$, with  Zipf exponent $\gamma=0.8$.
We consider  the wireless parameters and propagation models in Table~\ref{tab},
which are from the 3GPP release.
Hence, once the location $\xi$ of a
user has been given,
the SNR of the user from the macro BS
 $\text{SNR}_0(\xi)$  and from its pico BS $l$
 $\text{SNR}_l(\xi)$ can be determined using the cell geometry. Given these
SNRs, we obtain the macro and pico rates (in file/sec/Hz) using
Shannon's formula: $R_0(\xi)=\frac{1}{D}\log_2(1+\text{SNR}_0(\xi))$ and $R_l(\xi)=\frac{1}{D}\log_2(1+\text{SNR}_l(\xi))$.
  Similarly, given the location  of pico BS $l$,  the SNR of pico BS $l$ from the macro BS  $\text{SNR}_{0l}$ can be determined, based on which we can obtain the backhaul rate  $B_l=\frac{1}{D}\log_2(1+\text{SNR}_{0l})$.
%Note that a small $\gamma$ means a heavy-tail popularity distribution.
%File requests arrive as a Poisson process with arrival rate $\lambda_s=1$ file/sec.

We use Monte-Carlo simulation to estimate all relevant integrals.
Specifically, the whole macro coverage area is randomly
sampled with $2\times10^5$ points using the hotspot probabilities.
The macro and pico rates for each sampling position are then
calculated accordingly. Then, the locations of the arrivals are chosen independently and uniformly from these sampling positions. All the integrals are estimated using corresponding summations.
%We first estimate the rate-ratio threshold
%functions, ?(f), defined in (17) with the traffic load fixed by
%setting the arrival rate ?S = 1. We then place the mobiles in
%C in descending rate-ratio order ?(?1) ³ ?(?2) ³ ?(?3), . . .,
%once and for all. For any given value of ?, the needed ABS
%time f is estimated by summing the file transmission times
%of all mobiles with rate-ratio greater or equal to ? and then
%normalizing by N, i.e., f Å 1
%N
%
%i:?(?i)³?
%D
%R(?i) .
To avoid superfluous computation, for any given $W$ and $C_l$, $f$ is calculated using a fine grid of
$\rho_l$ values according to \eqref{eqn:rho-equation}, based on which we obtain $\rho_l(f)$.
%Finally,  smooth curves are  obtained using spline interpolation.

%These results were obtained for each pico cell
% = 1, 2, 3 and for both the above two cases. Results are
%graphed in Fig. 3 for the no interference case only. (The results
%for the full interference case are similar.) As can be seen, the
%three curves are smooth and strictly monotonically decreasing
%with maximum possible rate-ratio achieved at f = 0, reaching
%0 at some finite value of f. This value f is the maximum ABS
%time needed to clear all mobiles from the given pico at a traffic
%load of ?S = 1 mobiles/sec.

Fig.~\ref{fig:asym} (a) and Fig.~\ref{fig:sym} (a)  illustrate the threshold $\rho_l(f)$ versus pico time $f$ for the heterogenous and homogenous scenarios, respectively.
We can see that   $\rho_l(f)$  is strictly monotonically decreasing with maximum possible threshold achieved at $f=0$, reaching 0 at some finite value of $\bar f_l$. In addition, in the homogenous scenario, the three threshold curves coincide.
From Fig.~\ref{fig:asym} (a) (Fig.~\ref{fig:sym} (a)), we can immediately obtain the first figure in Fig.~\ref{fig:asym} (b)   (Fig.~\ref{fig:sym} (b)), which enables us to determine the unique optimal $f^*$ using Lemma~\ref{Lem:opt-f}. The second figure in Fig.~\ref{fig:asym} (b) (Fig.~\ref{fig:sym} (b)) also illustrates the unique optimal $f^*$ at which the minimum of $\tau(f)$ is achieved. The two optimal values of $f^*$ from the two figures coincide, illustrating Lemma~\ref{Lem:opt-f}. In addition, from Fig.~\ref{fig:sym} (b),  we can see that $\rho_l(f^*)=\frac{1}{L}=1/3$ in the homogenous scenario, which illustrating Corollary~\ref{Cor:opt-f-sym}.
Fig.~\ref{fig:WC} illustrates the optimal average time $\tau^*(W,\mathbf C)$ versus $W$ and $\mathbf C$. We can observe that as $W$ or $\mathbf C$ increases, the minimum  total  time decreases.  This illustrates Theorem~\ref{Thm:W-C-time}. In addition, we can see that when $W=1$ MHz, caching the 200 most popular files leads to a minimum total time reduction of $(0.2786-0.2059)/0.2786=26.1\%$ compared to $C_l=0$; when $C_l=0$,  increasing bandwidth from 1MHz to  1.4 MHz   leads to a minimum total time reduction of $(0.2786-0.1990)/0.2786=28.6\%$. Therefore, caching 200 files results in the performance improvement close to that offered  by using 0.4 MHz extra bandwidth. This demonstrates the  effectiveness of caching in Hetnets.

\begin{figure}
\begin{center}
  \subfigure[\small{$\rho_l(f)$ versus $f$.}]
  {\resizebox{4.5cm}{!}{\includegraphics{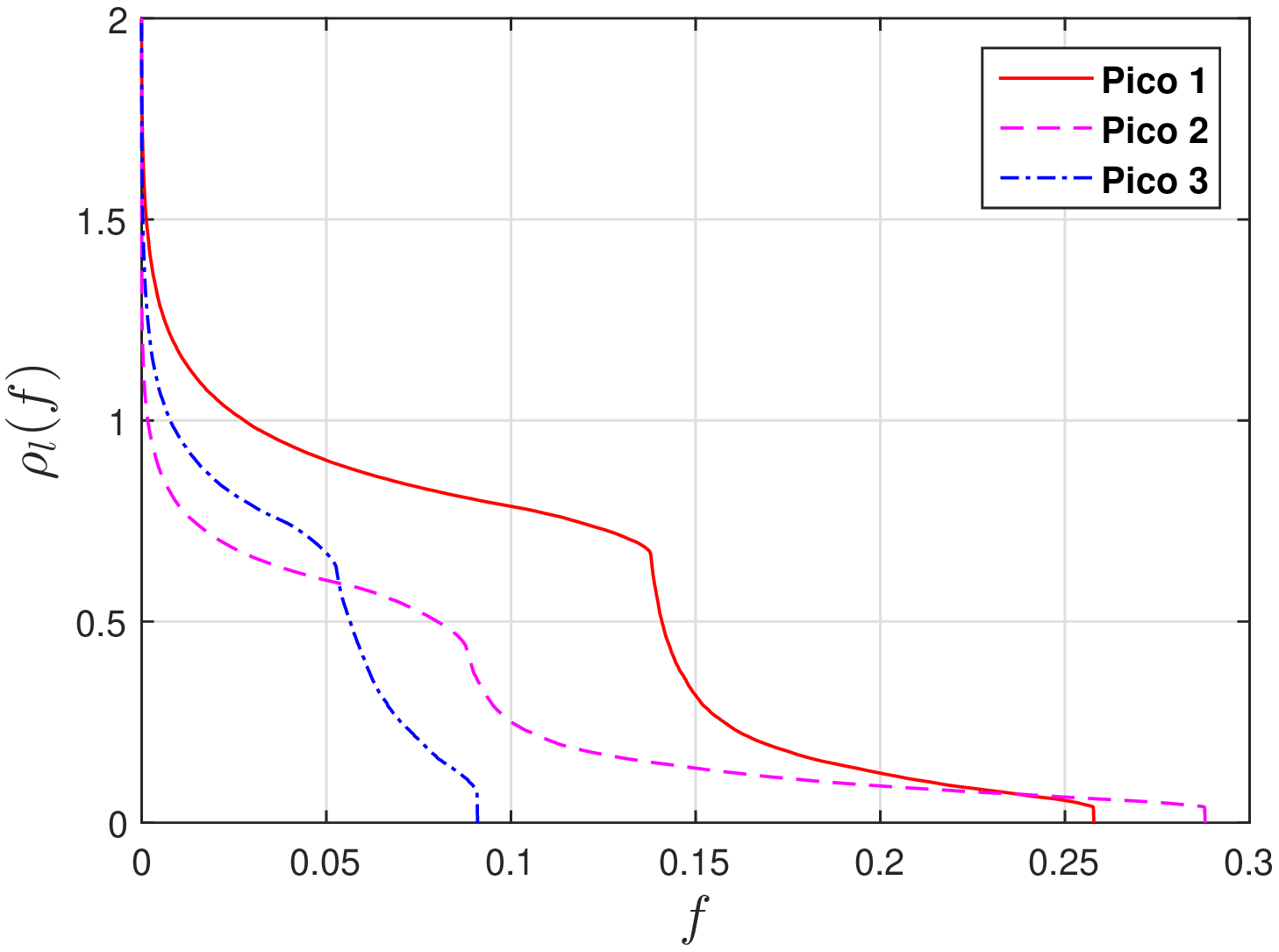}}}
\quad
  \subfigure[\small{$\sum_{l=1}^L\rho_l(f)$ and $\tau(f)$ versus $f$.}]
  {\resizebox{4.8cm}{!}{\includegraphics{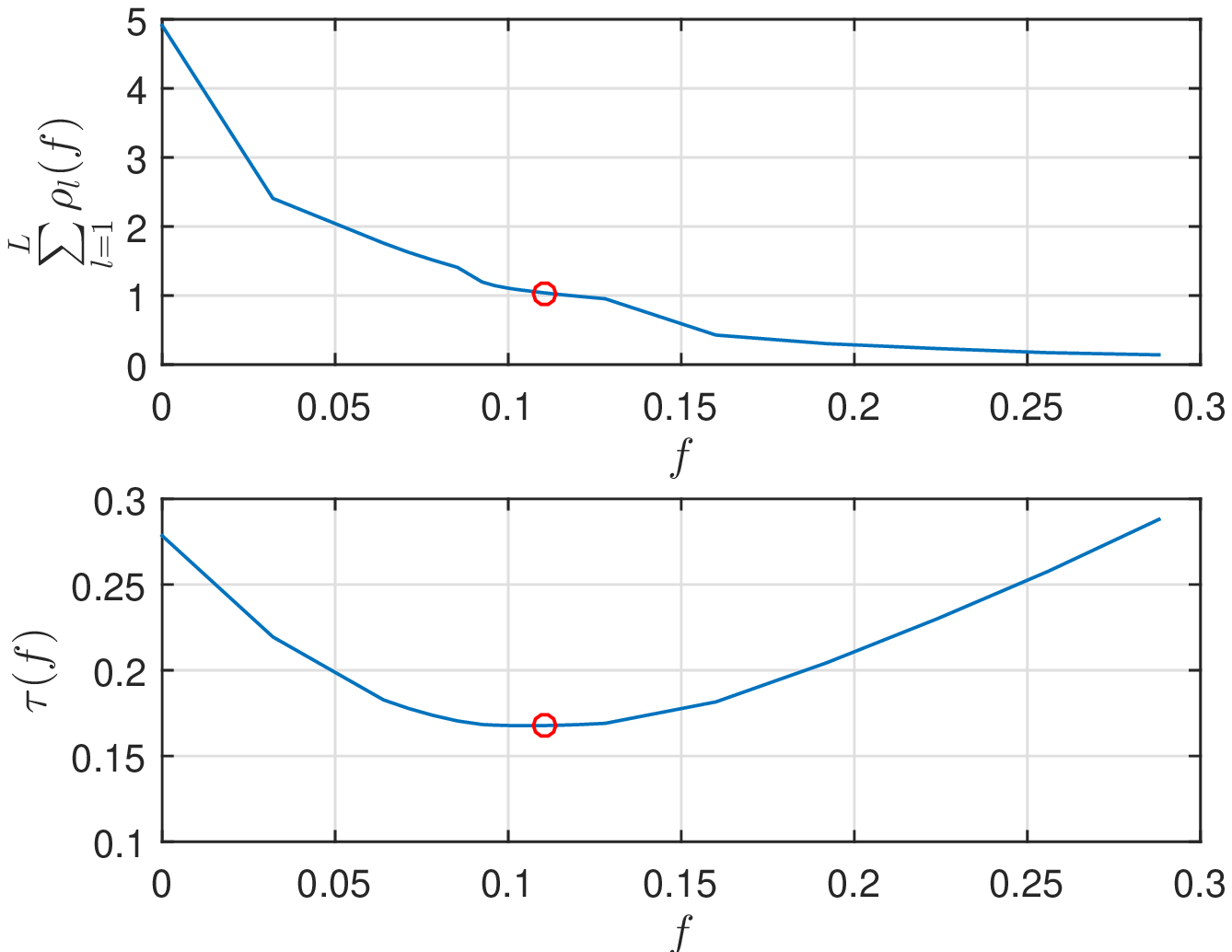}}}
  \end{center}
    \caption{\small{Heterogeneous scenario.  $W=1$ MHz and $C_l=200$, $l\in \mathcal L$.}}
\label{fig:asym}
\end{figure}

\begin{figure}
\begin{center}
  \subfigure[\small{$\rho_l(f)$ versus $f$.}]
  {\resizebox{4.5cm}{!}{\includegraphics{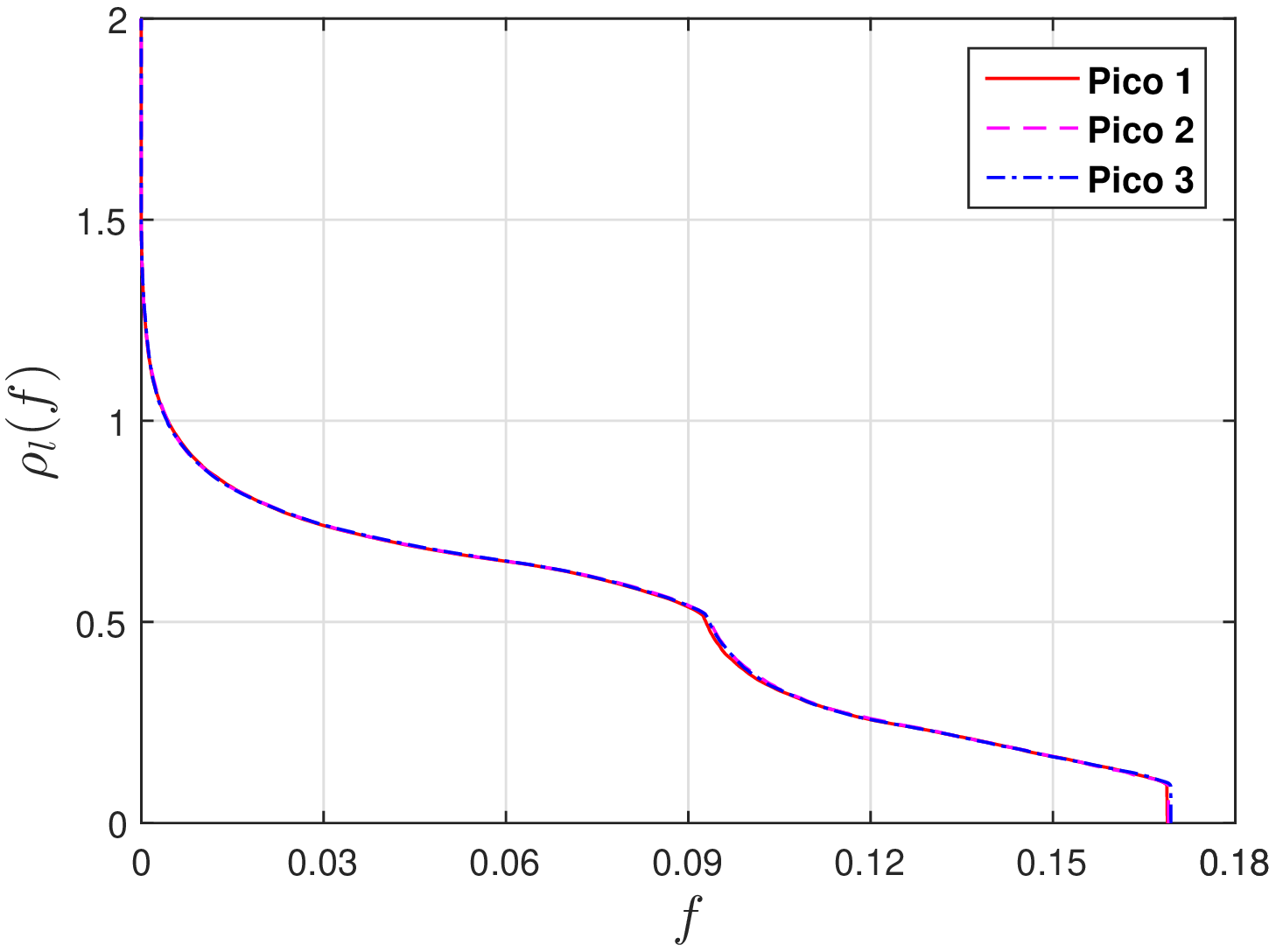}}}
\quad
  \subfigure[\small{$\sum_{l=1}^L\rho_l(f)$ and $\tau(f)$ versus $f$.}]
  {\resizebox{4.8cm}{!}{\includegraphics{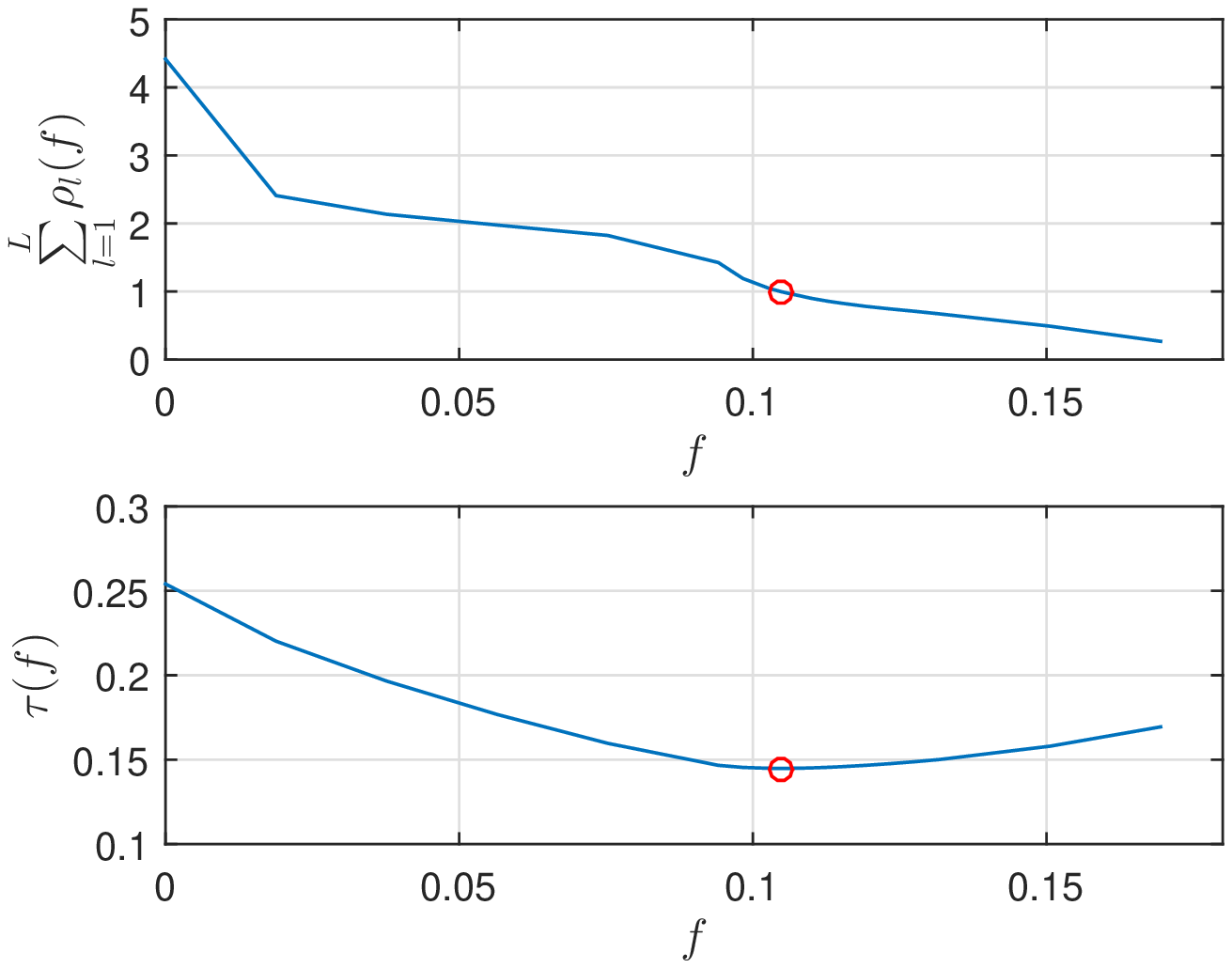}}}
  \end{center}
    \caption{\small{Homogenous scenario. $W=1$ MHz and  $C_l=200$, $l\in \mathcal L$.}}
\label{fig:sym}
\end{figure}

\begin{figure}
\begin{center}
 \includegraphics[width=5.7cm]{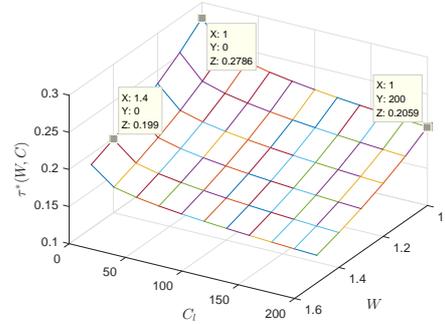}
  \end{center}
    \caption{\small{Minimum total time versus $W$ (MHz) and $C_l$ (files).  $C_l$, $l\in \mathcal L$ are the same.}}
\label{fig:WC}
\end{figure}

\section*{Appendix A: Proof of Lemma~\ref{Lem:solution-f}}

We shall solve a relaxed version of Problem~\ref{prob:subprob} where \eqref{eqn:s-constr} is replaced with 
\begin{align}
s_{l,n}\in[0,1],\quad   n\in \mathcal N.\label{eqn:s-constr-cont}
\end{align}
We shall show that the optimal solution to the relaxed problem satisfies \eqref{eqn:s-constr}, and hence is also the optimal solution to Problem~\ref{prob:subprob}.

Let $\rho_l$ denote the Lagrangian multiplier w.r.t. \eqref{eqn:pico-f-constr}. The corresponding Lagrangian is given by
\begin{align}
&\mathcal L(\{\mathbf x_l(\xi)\},\mathbf{s}_l, \rho_l)\nonumber\\
=& \sum_{n=1}^N\int_{\mathcal K_l}\frac{1-x_{l,n}(\xi)}{WR_0(\xi)}p_{n} \lambda(d\xi)\nonumber\\
&+ \sum_{n=1}^N\int_{\mathcal K_l}\frac{(1-s_{l,n})x_{l,n}(\xi)}{WB_l}p_{n} \lambda(d\xi)\nonumber\\
&+\rho_l\left(\sum_{n=1}^N\int_{\mathcal K_l}\frac{x_{l,n}(\xi)}{WR_l(\xi)}p_{n} \lambda(d\xi)-f \right)\nonumber\\
=& \sum_{n=1}^N\int_{\mathcal K_l}\left(\frac{\rho_l}{WR_l(\xi)}+\frac{1-s_{l,n}}{WB_l}-\frac{1}{WR_0(\xi)}\right)x_{l,n}(\xi)p_{n} \lambda(d\xi)\nonumber\\
&+  \sum_{n=1}^N\int_{\mathcal K_l}\frac{1}{WR_0(\xi)}p_{n} \lambda(d\xi)-\rho_lf\nonumber
\end{align}
where  $\{\mathbf x_l(\xi)\}$ and $\mathbf s_l$ satisfy \eqref{eqn:x-constr}, \eqref{eqn:s-constr-cont}, and \eqref{eqn:cache-constr}.  Minimizing the Lagrangian w.r.t. $\{\mathbf x_l(\xi)\}$ subject to \eqref{eqn:x-constr}, we have
\begin{align}
x_{l,n}(\xi, s_{l,n}, \rho_l) =&
\begin{cases}1, & \rho_{l}(\xi,s_{l,n})> \rho_l\\
0, & \text{otherwise}
\end{cases}
\end{align}
and
\begin{align}
&\min_{\{\mathbf x_l(\xi):\eqref{eqn:x-constr} \}}\mathcal L(\{\mathbf x_l(\xi)\},\mathbf s_l, \rho_l)\nonumber\\
=& \sum_{n=1}^Np_{n} \int_{\mathcal A_l(\rho_l,s_{l,n})}\left(\frac{\rho_l}{WR_l(\xi)}+\frac{1-s_{l,n}}{WB_l}-\frac{1}{WR_0(\xi)}\right)\lambda(d\xi)\nonumber\\
&+ \sum_{n=1}^N\int_{\mathcal K_l}\frac{1}{WR_0(\xi)}p_{n} \lambda(d\xi)-\rho_lf.\label{eqn:min-L-x}
\end{align}
Denote $a(s)\triangleq \int_{\mathcal A_l(\rho_l,s)}\left(\frac{\rho_l}{WR_l(\xi)}+\frac{1-s}{WB_l}-\frac{1}{WR_0(\xi)}\right)\lambda(d\xi)$. Now we show that  $a(s)$ is decreasing in $s$. Suppose $s>s'\geq 0$. Then, we have $\mathcal A_l(\rho_l,s')\subset\mathcal A_l(\rho_l,s)$, implying $\mathcal A_l(\rho_l,s)-\mathcal A_l(\rho_l,s')\subset \mathcal A_l(\rho_l,s)$. Thus, we have
\begin{align}
&a(s)-a(s')\nonumber\\
=&\int_{\mathcal A_l(\rho_l,s)-\mathcal A_l(\rho_l,s')}\left(\frac{\rho_l}{WR_l(\xi)}+\frac{1-s}{WB_l}-\frac{1}{WR_0(\xi)}\right)\lambda(d\xi)\nonumber\\
&-\int_{\mathcal A_l(\rho_l,s')}\left(\frac{s-s'}{WB_l}\right)\lambda(d\xi)<0\nonumber
\end{align}
as $\frac{\rho_l}{WR_l(\xi)}+\frac{1-s}{WB_l}-\frac{1}{WR_0(\xi)}<0$ for $\xi\in \mathcal A_l(\rho_l,s)$ and $s-s'>0$.
In addition, since $p_{n}$ is decreasing in $n$, further minimizing \eqref{eqn:min-L-x} w.r.t. $\mathbf s_l$ subject to \eqref{eqn:s-constr-cont} and \eqref{eqn:cache-constr}, we can obtain \eqref{eqn:opt-s} and the dual function
\begin{align}
&g_l(\rho_l,f)\nonumber\\
=&\min_{\{\mathbf x_l(\xi):\eqref{eqn:x-constr} \}, \{\mathbf s_l:\eqref{eqn:s-constr-cont}, \eqref{eqn:cache-constr}\}}\mathcal L(\{\mathbf x_l(\xi)\},\mathbf s_l, \rho_l)\nonumber\\
=&\sum_{n=1}^{C_l}\int_{\mathcal A_l(\rho_l,1)}\left(\frac{\rho_l}{WR_l(\xi)}-\frac{1}{WR_0(\xi)}\right)p_{n} \lambda(d\xi)\nonumber\\
&+\sum_{n=C_l+1}^N\int_{\mathcal A_l(\rho_l,0)}\left(\frac{\rho_l}{WR_l(\xi)}+\frac{1}{WB_l}-\frac{1}{WR_0(\xi)}\right)p_{n} \lambda(d\xi)\nonumber\\
&+ \sum_{n=1}^N\int_{\mathcal K_l}\frac{1}{WR_0(\xi)}p_{n} \lambda(d\xi)-\rho_lf\label{eqn:dual-func}
\end{align}
which is convex in $\rho_l$. On differentiating w.r.t. $\rho_l$ under the integral sign, we have
\begin{align}
g_l'(\rho_l,f)=&\sum_{n=1}^{C_l}\int_{\mathcal A_l(\rho_l,1)}\frac{1}{WR_l(\xi)}p_{n} \lambda(d\xi)\nonumber\\
&+\sum_{n=C_l+1}^N\int_{\mathcal A_l(\rho_l,0)}\frac{1}{WR_l(\xi)}p_{n} \lambda(d\xi)-f
\end{align}
which is continuous in $\rho_l$.  Note that $g_l'(\rho_l,f)\in [-f, \bar{f_l}-f]$. The dual problem is given by
\begin{align}
\tau_l(f)=\max_{\rho_l\geq 0}g_l(\rho_l,f). 
\end{align}
Since the dual problem is convex, when $f\leq\bar{f_l}$, the optimal dual variable satisfies $g_l'(\rho_l,f)=0$. When $f>\bar{f_l}$, $g_l'(\rho_l,f)<0$, the optimal dual variable is 0. Thus, the optimal dual value is given by
\begin{align}
\tau_l(f)=&\int_{\mathcal K_l}\frac{1}{WR_0(\xi)} \lambda(d\xi)-S_{C_l}\int_{\mathcal A_l(\rho_l(f),1)}\frac{1}{WR_0(\xi)} \lambda(d\xi)\nonumber\\
&-(1-S_{C_l})\int_{\mathcal A_l(\rho_l(f),0)}\left(\frac{1}{WR_0(\xi)} - \frac{1}{WB_l}\right)\lambda(d\xi)\nonumber
\end{align}
which is equal to \eqref{eqn:subp-opt}.
 The corresponding $\{\mathbf x_l^*(\xi)\}$ given by \eqref{eqn:opt-x-rho} (with $\mathbf s_l$ given by \eqref{eqn:opt-s}) is feasible. Therefore, by the Lagrange Sufficiency Theorem, \eqref{eqn:opt-x-rho} and \eqref{eqn:opt-s} are the optimal solution to the relaxed  version of Problem~\ref{prob:subprob} for any given $f\geq 0$. Note that  $\mathbf s_l^*$ given by \eqref{eqn:opt-s} satisfies \eqref{eqn:s-constr}. Thus,\eqref{eqn:opt-x-rho} and   \eqref{eqn:opt-s} are also  the optimal solution to Problem~\ref{prob:subprob} for any given $f\geq 0$. As $x_{l,n}^*(\xi)$ is a  function of $f$, we also write it as  $x_{l,n}^*(\xi, f)$.

\section*{Appendix B: Proof of Lemma~\ref{Lem:opt-f}}

Given $f\in(0, \bar{f_l})$, for suitable $\delta f>0$, we have
\begin{align}
&\tau_l(f)-\tau_l(f-\delta f)\nonumber\\
=&
-S_{C_l}\int_{\substack{\mathcal A_l(\rho_l(f),1)\\-\mathcal A_l(\rho_l(f-\delta f),1)}}\frac{1}{WR_0(\xi)}\lambda(d\xi)\nonumber\\
&-(1-S_{C_l})\int_{\substack{\mathcal A_l(\rho_l(f),0)\\-\mathcal A_l(\rho_l(f-\delta f),0)}}\left(\frac{1}{WR_0(\xi)}-\frac{1}{WB_l}\right) \lambda(d\xi)\nonumber\\
=&-S_{C_l}\int_{\{\xi\in \mathcal K_l:\rho(\xi,1)\in [\rho_l(f),\rho_l(f-\delta f)] \}}\frac{\rho(\xi,1)}{WR_l(\xi)}\lambda(d\xi)\nonumber\\
&-(1-S_{C_l})\int_{\{\xi\in \mathcal K_l:\rho(\xi,0)\in [\rho_l(f),\rho_l(f-\delta f)] \}}\frac{\rho(\xi,0)}{WR_l(\xi)}\lambda(d\xi)\nonumber\\
\stackrel{(a)}{=}&-\tilde{\rho_l}S_{C_l}\int_{\{\xi\in \mathcal K_l:\rho(\xi,1)\in [\rho_l(f),\rho_l(f-\delta f)] \}}\frac{1}{WR_l(\xi)}\lambda(d\xi)\nonumber\\
&-\hat{\rho_l}(1-S_{C_l})\int_{\{\xi\in \mathcal K_l:\rho(\xi,0)\in [\rho_l(f),\rho_l(f-\delta f)] \}}\frac{1}{WR_l(\xi)}\lambda(d\xi)\nonumber\\
=&-\tilde{\rho_l}S_{C_l}\Bigg(\int_{\mathcal A_l(\rho_l(f),1)}\frac{1}{WR_l(\xi)}\lambda(d\xi)\nonumber\\
&\hspace{15mm} -\int_{\mathcal A_l(\rho_l(f-\delta f),1)}\frac{1}{WR_l(\xi)}\lambda(d\xi)\Bigg)\nonumber\\
&-\tilde{\rho_l}(1-S_{C_l})\Bigg(\int_{\mathcal A_l(\rho_l(f),0)}\frac{1}{WR_l(\xi)}\lambda(d\xi)\nonumber\\
&\hspace{22mm}-\int_{\mathcal A_l(\rho_l(f-\delta f),0)}\frac{1}{WR_l(\xi)}\lambda(d\xi)\Bigg)+\Delta(\delta f)\nonumber
\end{align}
\begin{align}
=&-\tilde{\rho_l}\Bigg(S_{C_l}\int_{\mathcal A_l(\rho_l(f),1)}\frac{1}{WR_l(\xi)} \lambda(d\xi)\nonumber\\
&\hspace{10mm}+(1-S_{C_l})\int_{\mathcal A_l(\rho_l(f),0)}\frac{1}{WR_l(\xi)}\lambda(d\xi)\Bigg)\nonumber\\
&+\tilde{\rho_l}\Bigg(S_{C_l}\int_{\mathcal A_l(\rho_l(f-\delta f),1)}\frac{1}{WR_l(\xi)} \lambda(d\xi)\nonumber\\
&\hspace{10mm}+(1-S_{C_l})\int_{\mathcal A_l(\rho_l(f-\delta f),0)}\frac{1}{WR_0(\xi)} \lambda(d\xi)\Bigg)+\Delta(\delta f)\nonumber\\
=&-\tilde{\rho_l}f+\tilde{\rho_l}(f-\delta f)+\Delta(\delta f)=-\tilde{\rho_l}\delta f+\Delta(\delta f)\label{eqn:proof-dif-tau}
\end{align}
where (a) is due to mean value theorem and $\tilde{\rho_l}, \hat{\rho_l}\in [\rho_l(f), \rho_l(f-\delta f)]$, and $\Delta (\delta f)$ is defined as
\begin{align}
&\Delta (\delta f)\nonumber\\
\triangleq& 
%(\tilde{\rho_l}-\hat{\rho_l})(1-S_{C_l})\Bigg(\int_{\mathcal A_l(\rho_l(f),0)}\frac{1}{WR_l(\xi)}\lambda(d\xi)\nonumber\\
%&-\int_{\mathcal A_l(\rho_l(f-\delta f),0)}\frac{1}{WR_l(\xi)}\lambda(d\xi)\Bigg)\nonumber\\
%=&
 (\tilde{\rho_l}-\hat{\rho_l})(1-S_{C_l}) \int_{\mathcal A_l(\rho_l(f),0)-\mathcal A_l(\rho_l(f-\delta f),0)}\frac{1}{WR_l(\xi)}\lambda(d\xi)
\end{align}
Since $\tilde{\rho_l}, \hat{\rho_l}\in [\rho_l(f), \rho_l(f-\delta f)]$, we have $|\tilde{\rho_l}-\hat{\rho_l}|\leq  \rho_l(f-\delta f)-\rho_l(f)$. As $\rho_l(f)$ has a finite unique derivative for almost all $f$, we have $ \rho_l(f-\delta f)-\rho_l(f)\leq C\delta f$ for some finite positive $C$. Thus, we can show $\Delta (\delta f)/\delta f\to 0$ as $\delta f\to 0$.
Dividing $\delta f$ on both sides of \eqref{eqn:proof-dif-tau} and taking $\delta f\to 0$, we can show $\tau_l'(f)=-\rho_l(f)$ for all $f\in (0,\bar f_l)$. On the other hand, when $f\in (f_l,\bar f)$,  we have $\tau_l'(f)=0$ by Corollary~\ref{Cor:solu-f0fbar} and $\rho_l(f)=0$ by Lemma~\ref{Lem:solution-f}. Thus, when $f\in (f_l,\bar f)$, we can show  $\tau_l'(f)=-\rho_l(f)$. Therefore,  we have $\tau'(f)=1-\sum_{l=1}^L \rho_l(f)$ for all $f\in(0,\bar f)$. 
If $\sum_{l=1}^L \bar{\rho_l}\leq 1$, then $\sum_{l=1}^L \rho_l(f)< 1$ for all $f\in[0,\bar f]$, as $\sum_{l=1}^L \rho_l(f)<\sum_{l=1}^L \bar{\rho_l}$ for all $f\in(0,\underline f)$. Thus, if $\sum_{l=1}^L \bar{\rho_l}\leq 1$,  we have $\tau'(f)>0$ for all $f\in(0,\bar f)$, implying $f^*=0$. If $\sum_{l=1}^L \rho_l(\bar f)\geq 1$, then $\sum_{l=1}^L \rho_l(f)\geq 1$ for all $f\in(0,\bar f)$, as $\rho_l(f)$ is non-increasing in $f$. Thus, if $\sum_{l=1}^L \rho_l(\bar f)\geq 1$, we have $\tau'(f)<0$ for all $f\in(0,\bar f)$, implying $f^*=\bar f$. Otherwise, $f^*\in(0,\bar f)$. As $\tau(f)$ is convex in $f$, we know that $\tau'(f^*)=0$.

\section*{Appendix C: Proof of Lemma~\ref{Lem:derivative}}

Note that $\rho_l(f,W,S_{C_l})$ increases in $S_{C_l}$ and decreases in  $W$.
First, we show \eqref{eqn:partial-W}.
\begin{align}
&\tau_l(S_{C_l}, W+\delta W,f)-\tau_l(S_{C_l}, W,f)\nonumber\\
=&\int_{\mathcal K_l}\frac{1}{(W+\delta W)R_0(\xi)} \lambda(d\xi)\nonumber\\
&-S_{C_l}\int_{\mathcal A_l(\rho_l(S_{C_l}, W+\delta W,f),1)}\frac{1}{(W+\delta W)R_0(\xi)} \lambda(d\xi)\nonumber\\
&-(1-S_{C_l})\int_{\mathcal A_l(\rho_l(S_{C_l}, W+\delta W,f),0)}\Bigg(\frac{1}{(W+\delta W)R_0(\xi)} \nonumber\\
&\hspace{55mm}- \frac{1}{(W+\delta W)B_l}\Bigg)\lambda(d\xi)\nonumber\\
&-\int_{\mathcal K_l}\frac{1}{WR_0(\xi)} \lambda(d\xi)+S_{C_l}\int_{\mathcal A_l(\rho_l(S_{C_l}, W,f),1)}\frac{1}{WR_0(\xi)} \lambda(d\xi)\nonumber\\
&+(1-S_{C_l})\int_{\mathcal A_l(\rho_l(S_{C_l}, W,f),0)}\left(\frac{1}{WR_0(\xi)} - \frac{1}{WB_l}\right)\lambda(d\xi)\nonumber\\
=&-\frac{\delta W}{W}\int_{\mathcal K_l}\frac{1}{(W+\delta W)R_0(\xi)} \lambda(d\xi)\nonumber\\
&+\frac{\delta W}{W}S_{C_l}\int_{\mathcal A_l(\rho_l(S_{C_l}, W+\delta W,f),1)}\frac{1}{(W+\delta W)R_0(\xi)} \lambda(d\xi)\nonumber\\
&+\frac{\delta W}{W}(1-S_{C_l})\int_{\mathcal A_l(\rho_l(S_{C_l}, W+\delta W,f),0)}\Bigg(\frac{1}{(W+\delta W)R_0(\xi)} \nonumber\\
&\hspace{55mm}- \frac{1}{(W+\delta W)B_l}\Bigg)\lambda(d\xi)\nonumber\\
&-S_{C_l}\int_{\substack{\mathcal A_l(\rho_l(S_{C_l}, W+\delta W,f),1)\\-\mathcal A_l(\rho_l(S_{C_l}, W,f),1)}}\frac{1}{WR_0(\xi)} \lambda(d\xi)\nonumber\\
&-(1-S_{C_l})\int_{\substack{\mathcal A_l(\rho_l(S_{C_l}, W+\delta W,f),0)\\-\mathcal A_l(\rho_l(S_{C_l}, W,f),0)}}\left(\frac{1}{WR_0(\xi)} - \frac{1}{WB_l}\right)\lambda(d\xi)\nonumber
\end{align}
\begin{align}
\stackrel{(a)}{=}&-\frac{\delta W}{W}\tau_l(S_{C_l}, W+\delta W,f)\nonumber\\
&-S_{C_l}\tilde{\rho_l}\int_{\substack{\mathcal A_l(\rho_l(S_{C_l}, W+\delta W,f),1)\\-\mathcal A_l(\rho_l(S_{C_l}, W,f),1)}}\frac{1}{WR_l(\xi)} \lambda(d\xi)\nonumber\\
&-(1-S_{C_l})\hat{\rho_l}\int_{\substack{\mathcal A_l(\rho_l(S_{C_l}, W+\delta W,f),0)\\-\mathcal A_l(\rho_l(S_{C_l}, W,f),0)}}\frac{1}{WR_l(\xi)} \lambda(d\xi)\nonumber\\
=&-\frac{\delta W}{W}\tau_l(S_{C_l}, W+\delta W,f)-\tilde{\rho_l}(f-f)-\frac{\delta W}{W}\tilde{\rho_l}f\nonumber\\
&-\Delta (\delta W)\label{eqn:proof-dif-W}
\end{align}
where (a) is due to mean value theorem and $\tilde{\rho_l}, \hat{\rho_l}\in [\rho_l(S_{C_l}, W+\delta W,f), \rho_l(S_{C_l}, W,f)]$, and $\Delta (\delta f)$ is defined as
\begin{align}
&\Delta (\delta W)\nonumber\\
\triangleq&(1-S_{C_l})(\hat{\rho_l}-\tilde{\rho_l})\int_{\substack{\mathcal A_l(\rho_l(S_{C_l}, W+\delta W,f),0)\\-\mathcal A_l(\rho_l(S_{C_l}, W,f),0)}}\frac{1}{WR_l(\xi)} \lambda(d\xi)\nonumber
\end{align}
Similarly, we can show $\Delta (\delta W)/\delta W\to 0$ as $\delta W\to 0$.
Dividing $\delta W$ on both sides of \eqref{eqn:proof-dif-W} and taking $\delta W\to 0$, we can show \eqref{eqn:partial-W}.

Next, we show \eqref{eqn:partial-S}.
\begin{align}
&\tau_l(S_{C_l}+\delta S_{C_l},W,f)-\tau_l(f,W,S_{C_l})\nonumber\\
=&-(S_{C_l}+\delta S_{C_l})\int_{\mathcal A_l(\rho_l(S_{C_l}+\delta S_{C_l},W,f),1)}\frac{1}{WR_0(\xi)} \lambda(d\xi)\nonumber\\
&-(1-S_{C_l}-\delta S_{C_l})\int_{\mathcal A_l(\rho_l(S_{C_l}+\delta S_{C_l},W,f),0)}\left(\frac{1}{WR_0(\xi)}-\frac{1}{WB_l} \right)\lambda(d\xi)\nonumber\\
&+S_{C_l}\int_{\mathcal A_l(\rho_l(f,W,S_{C_l}),1)}\frac{1}{WR_0(\xi)} \lambda(d\xi)\nonumber\\
&+(1-S_{C_l})\int_{\mathcal A_l(\rho_l(f,W,S_{C_l}),0)}\left(\frac{1}{WR_0(\xi)}-\frac{1}{WB_l} \right)\lambda(d\xi)\nonumber\\
=&S_{C_l}\int_{\substack{\mathcal A_l(\rho_l(f,W,S_{C_l}),1)\\-\mathcal A_l(\rho_l(S_{C_l}+\delta S_{C_l},W,f),1)}}\frac{1}{WR_0(\xi)} \lambda(d\xi)\nonumber\\
&+(1-S_{C_l})\int_{\substack{\mathcal A_l(\rho_l(f,W,S_{C_l}),0)\\-\mathcal A_l(\rho_l(S_{C_l}+\delta S_{C_l},W,f),0)}}\left(\frac{1}{WR_0(\xi)}-\frac{1}{WB_l} \right)\lambda(d\xi)\nonumber\\
&-\delta S_{C_l}\Bigg(\int_{\mathcal A_l(\rho_l(S_{C_l}+\delta S_{C_l},W,f),1)}\frac{1}{WR_0(\xi)} \lambda(d\xi)\nonumber\\
&\hspace{5mm}-\int_{\mathcal A_l(\rho_l(S_{C_l}+\delta S_{C_l},W,f),0)}\left(\frac{1}{WR_0(\xi)}-\frac{1}{WB_l} \Bigg)\lambda(d\xi)\right)\nonumber
\end{align}
\begin{align}
\stackrel{(a)}{=}&\tilde{\rho_l}S_{C_l}\int_{\substack{\mathcal A_l(\rho_l(f,W,S_{C_l}),1)\\-\mathcal A_l(\rho_l(S_{C_l}+\delta S_{C_l},W,f),1)}}\frac{1}{WR_l(\xi)} \lambda(d\xi)\nonumber\\
&+\hat{\rho_l}(1-S_{C_l})\int_{\substack{\mathcal A_l(\rho_l(f,W,S_{C_l}),0)\\-\mathcal A_l(\rho_l(S_{C_l}+\delta S_{C_l},W,f),0)}}\frac{1}{WR_l(\xi)}\lambda(d\xi)\nonumber\\
&-\delta S_{C_l}\Bigg(\int_{\mathcal A_l(\rho_l(S_{C_l}+\delta S_{C_l},W,f),1)}\frac{1}{WR_0(\xi)} \lambda(d\xi)\nonumber\\
&-\int_{\mathcal A_l(\rho_l(S_{C_l}+\delta S_{C_l},W,f),0)}\left(\frac{1}{WR_0(\xi)}-\frac{1}{WB_l} \Bigg)\lambda(d\xi)\right)\nonumber\\
=&\tilde{\rho_l}S_{C_l}\int_{\mathcal A_l(\rho_l(f,W,S_{C_l}),1)}\frac{1}{WR_l(\xi)} \lambda(d\xi)\nonumber\\
&-\tilde{\rho_l}(S_{C_l}+\delta S_{C_l})\int_{\mathcal A_l(\rho_l(S_{C_l}+\delta S_{C_l},W,f),1)}\frac{1}{WR_l(\xi)} \lambda(d\xi)\nonumber\\
&+\hat{\rho_l}(1-S_{C_l})\int_{\mathcal A_l(\rho_l(f,W,S_{C_l}),0)}\frac{1}{WR_l(\xi)}\lambda(d\xi)\nonumber\\
&-\hat{\rho_l}(1-S_{C_l}-\delta S_{S_{C_l}})\int_{\mathcal A_l(\rho_l(S_{C_l}+\delta S_{C_l},W,f),0)}\frac{1}{WR_l(\xi)}\lambda(d\xi)\nonumber\\
&+\Delta_1(\delta S_{C_l})\nonumber\\
=&\tilde{\rho_l}(f-f)+\Delta_1(\delta S_{C_l})+\Delta_2(\delta S_{C_l})\nonumber\\
=&\Delta_1(\delta S_{C_l})+\Delta_2(\delta S_{C_l})\nonumber
\end{align}
where (a) is due to mean value theorem and $\tilde{\rho_l}, \hat{\rho_l}\in [\rho_l(f,W,S_{C_l}), \rho_l(S_{C_l}+\delta S_{C_l},W,f)]$, and $\Delta_1(\delta S_{C_l})$ and  $\Delta_2(\delta S_{C_l})$ are defined as
\begin{align}
&\Delta_1(\delta S_{C_l})\nonumber\\
\triangleq &\tilde{\rho_l}\delta S_{C_l}\int_{\mathcal A_l(\rho_l(S_{C_l}+\delta S_{C_l},W,f),1)}\frac{1}{WR_l(\xi)} \lambda(d\xi)\nonumber\\
&-\hat{\rho_l}\delta S_{S_{C_l}}\int_{\mathcal A_l(\rho_l(S_{C_l}+\delta S_{C_l},W,f),0)}\frac{1}{WR_l(\xi)}\lambda(d\xi)\nonumber\\
&-\delta S_{C_l}\Bigg(\int_{\mathcal A_l(\rho_l(S_{C_l}+\delta S_{C_l},W,f),1)}\frac{1}{WR_0(\xi)} \lambda(d\xi)\nonumber\\
&-\int_{\mathcal A_l(\rho_l(S_{C_l}+\delta S_{C_l},W,f),0)}\left(\frac{1}{WR_0(\xi)}-\frac{1}{WB_l} \right)\lambda(d\xi)\Bigg)\nonumber\\
&\Delta_2(\delta S_{C_l})\nonumber\\
\triangleq& (\hat{\rho_l}-\tilde{\rho_l})(1-S_{C_l})\int_{\substack{\mathcal A_l(\rho_l(f,W,S_{C_l}),0)\\-\mathcal A_l(\rho_l(S_{C_l}+\delta S_{C_l},W,f),0)}}\frac{1}{WR_l(\xi)}\lambda(d\xi)\nonumber\\
&+(\hat{\rho_l}-\tilde{\rho_l})\delta S_{S_{C_l}}\int_{\mathcal A_l(\rho_l(S_{C_l}+\delta S_{C_l},W,f),0)}\frac{1}{WR_l(\xi)}\lambda(d\xi)
\end{align}
Similarly, we can show  $\Delta_2(\delta S_{C_l})/\delta S_{C_l}\to 0$ as $\delta S_{C_l}\to 0$. In addition, we have
\begin{align}
&\lim_{\delta S_{C_l}\to 0}\frac{\Delta_1(\delta S_{C_l})}{\delta S_{C_l}}\nonumber\\
=&\int_{\mathcal A_l(\rho_l(f,W,S_{C_l}),1)}\left(\frac{\rho_l(f,W,S_{C_l})}{WR_l(\xi)}-\frac{1}{WR_0(\xi)}  \right)\lambda(d\xi)\nonumber\\
&-\int_{\mathcal A_l(\rho_l(f,W,S_{C_l}),0)}\left(\frac{\rho_l(f,W,S_{C_l})}{WR_l(\xi)}-\left(\frac{1}{WR_0(\xi)}-\frac{1}{WB_l} \right)\right)\lambda(d\xi)\nonumber\\
=&\int_{\substack{\mathcal A_l(\rho_l(f,W,S_{C_l}),1)\\-\mathcal A_l(\rho_l(f,W,S_{C_l}),0)}}\left(\frac{\rho_l(f,W,S_{C_l})}{WR_l(\xi)}-\frac{1}{WR_0(\xi)}  \right)\lambda(d\xi)\nonumber\\
&-\int_{\mathcal A_l(\rho_l(f,W,S_{C_l}),0)}\frac{1}{WB_l} \lambda(d\xi)\nonumber\\
<&0
\end{align}
Thus, we can show \eqref{eqn:partial-S}.

\section*{Appendix D: Proof of Theorem~\ref{Thm:W-C-time}}
By Lemma~\ref{Lem:derivative}, we have $\tau(f,W,\mathbf C)\leq \tau(f,W',\mathbf C')$. Thus, we have $\tau(f,W',\mathbf C')\geq \tau(f^*(W',\mathbf C'),W',\mathbf C')\geq \tau(f^*(W',\mathbf C'),W,\mathbf C)\geq \tau(f^*(W,\mathbf C),W,\mathbf C)$. Note that $\tau^*(W',\mathbf C')\triangleq \tau(f^*(W',\mathbf C'),W',\mathbf C')$
 and $\tau^*(W,\mathbf C)\triangleq \tau(f^*(W,\mathbf C),W,\mathbf C)$. Therefore, we complete the proof.

\end{document}